\documentclass[a4paper,UKenglish,cleveref, autoref, thm-restate]{lipics-v2021}
\nolinenumbers
\usepackage[T1]{fontenc}
\usepackage{libertine}
\usepackage{tabularx}

\pdfoutput=1 
\hideLIPIcs  

\usepackage{algorithm2e}
\usepackage{tipa}
\usepackage{tikz}
\usepackage{apxproof}

 \usepackage[utf8]{inputenc}

\graphicspath{{fig/}}
\newcommand{\hide}[1]{}
\usepackage[colorinlistoftodos,prependcaption,textsize=tiny]{todonotes}

\theoremstyle{definition}

\DeclareMathOperator*{\argmax}{\textsc{}{argmax}}
\DeclareMathOperator*{\argmin}{\textsc{}{argmin}}
\newcommand{\Active}{\textsc{Active}}
\newcommand{\Above}{\textsc{Above}}
\newcommand{\Below}{\textsc{Below}}
\newcommand{\Contain}{\textsc{Contain}}
\newcommand{\Slab}{\textsc{Slab}}
\newcommand{\Subslab}{\textsc{Subslab}}
\newcommand{\Piece}{\textsc{Piece}}
\newcommand{\ub}{\textsc{Ub}}
\newcommand{\lb}{\textsc{Lb}}
\newcommand{\upb}{\textsc{Upb}}
\newcommand{\lpb}{\textsc{Lpb}}

\newcommand{\Strip}{\textsc{Strip}}

\newcommand{\OPT}{\textsc{opt}}

\newcommand{\corridor}{\textsc{Corridor}}

\def\cT{\mathcal{T}}
  \newcommand\upleft[1][1]{%
  \begin{tikzpicture}[scale=#1]
  \draw[line width=0.25mm] (0,0) -- (0,1) -- (-2,1);
  \end{tikzpicture}
  }
    \newcommand\leftup[1][1]{%
  \begin{tikzpicture}[scale=#1]
  \draw[line width=0.25mm] (0,0) -- (0,-1) -- (-2,-1);
  \end{tikzpicture}
  }

\title{A fast algorithm for computing a planar support for non-piercing rectangles}
\author{Ambar Pal}{Johns Hopkins University}{ambar@jhu.edu}{https://orcid.org/0000-0001-6775-4980}{}
\author{Rajiv Raman}{Indraprastha Institute of Information Technology Delhi}{rajiv@iiitd.ac.in}{https://orcid.org/0009-0000-8013-9421}{}
\author{Saurabh Ray}{NYU Abu Dhabi}{saurabh.ray@nyu.edu}{}{}
\author{Karamjeet Singh}{Indraprastha Institute of Information Technology Delhi}{karamjeets@iiitd.ac.in}{}{}
\authorrunning{A. Pal and R. Raman and S. Ray and K. Singh}
\Copyright{Ambar Pal and Rajiv Raman and Saurabh Ray and Karamjeet Singh} 



\category{} 

\relatedversion{} 

\begin{document}
\maketitle

\begin{abstract}
For a hypergraph $\mathcal{H}=(X,\mathcal{E})$
a \emph{support} is a graph $G$ on $X$ such that for each $E\in\mathcal{E}$, the induced
subgraph of $G$ on the elements in $E$ is connected. If $G$ is planar, we call it
a planar support. A set of axis parallel rectangles $\mathcal{R}$ forms a non-piercing family if for any $R_1, R_2 \in \mathcal{R}$, $R_1 \setminus R_2$ is connected.

Given a set $P$ of $n$ points in $\mathbb{R}^2$ and a set $\mathcal{R}$ of $m$ \emph{non-piercing} axis-aligned rectangles, we give an algorithm for computing a planar support
for the hypergraph $(P,\mathcal{R})$ in $O(n\log^2 n + (n+m)\log m)$ time, where each $R\in\mathcal{R}$ defines
a hyperedge consisting of all points of $P$ contained in~$R$.

 We use this result to show that if for a family of axis-parallel rectangles, any point in the plane is contained in at most $k$ pairwise \emph{crossing} rectangles (a pair of intersecting rectangles such that neither contains
 a corner of the other is called a crossing pair of rectangles), 
 then we can obtain a support as the union of $k$ planar graphs. 
\end{abstract}

\section{Introduction}
Given a hypergraph $\mathcal{H}=(V,\mathcal{E})$, 
a \emph{support} is a graph $G$ on $V$ such that for all $E\in\mathcal{E}$, the subgraph induced by $E$ in $G$, denoted $G[E]$ is connected.
The notion of a support was introduced by Voloshina and Feinberg \cite{voloshina1984planarity} in the context of VLSI circuits.
Since then, this notion has found wide applicability in several areas, such as visualizing hypergraphs 
\cite{bereg2015colored,bereg2011red,brandes2010blocks,brandes2012path,buchin2011planar,havet2022overlaying,hurtado2018colored},
in the design of networks \cite{anceaume2006semantic,baldoni2007tera,baldoni2007efficient,chand2005semantic,hosoda2012approximability,korach2003clustering,onus2011minimum}, and 
similar notions have been used in the analysis of local search algorithms for geometric problems \cite{Mustafa17,BasuRoy2018,Cohen-AddadM15,krohn2014guarding,mustafa2010improved,RR18}.

Any hypergraph clearly has a support: the complete graph on all vertices.
In most applications however, we require a support with an additional structure. For example, we may want a support
with the fewest number of edges, or a support that comes from a restricted family of graphs (e.g., outerplanar graphs).

Indeed, the problem of constructing a support has been studied by several research communities. 
For example, Du, et al., \cite{du1986optimization,du1995complexity,du1988matroids} studied the problem of minimizing the number of edges in a support, motivated
by questions in the design of vacuum systems. The problem has also been studied under the topic of ``minimum overlay networks''
\cite{hosoda2012approximability,chockler2007constructing} with applications to distributed computing. 
Johnson and Pollack \cite{johnson1987hypergraph} showed that it is NP-hard to decide if a hypergraph
admits a planar support. 

In another line of work, motivated by the analysis of approximation algorithms for packing and covering problems on \emph{geometric
hypergraphs}\footnote{In a geometric hypergraph, the elements of the hypergraph are points in the plane, 
and the hyperedges are defined by geometric regions in the plane, where each region defines a hyperedge consisting of 
all points contained in the region.}, several authors have considered the problem of constructing supports that belong to a family having sublinear sized separators\footnote{A family of graphs $\mathcal{G}$ admits sublinear sized
separators if there exist $0<\alpha,\beta<1$ s.t. for any $G\in \mathcal{G}$,
there exists a set $S\subseteq V(G)$ s.t. $G[V\setminus S]$ consists of two parts
$A$ and $B$ with $|A|,|B|\le \alpha |V|$, and 
there is no path in $G[V\setminus S]$ between a vertex in $A$ to a vertex in $B$. 
Further, $|S|\le |V|^{1-\beta}$.} such as planar graphs, or graphs
of bounded genus \cite{RR18,DBLP:journals/dcg/RamanR22,raman2024hypergraph}.
For this class of problems, the problem of interest is only to show the existence of a support from a restricted family
of graphs.

{\bf Our contribution.} 
So far there are very few tools or 
techniques to construct a support for a given hypergraph or even to show that a support with desired properties (e.g., planarity) exists.
Our paper presents a fast algorithm to construct a planar support for a restricted setting,
namely hypergraphs defined by axis-parallel rectangles that are \emph{non-piercing}, i.e., for each pair of intersecting
rectangles, one of them contains a corner of the other. This may seem rather restrictive. However, even if we allow each rectangle
to belong to at most one piercing pair of rectangles, it is not difficult to construct examples where for any $r\ge 3$,
any support must have $K_{r,r}$ as a topological minor. To see this, consider a geometric drawing of $K_{r,r}$ in the usual manner,
i.e., the two partite sets on two vertical lines, and the edges as straight-line segments. Replace each edge of the graph
by a long path, and then replace each edge along each path by a small rectangle that contains exactly two points. 
Where the edges cross, a pair of rectangles corresponding to each edge cross.
Since each rectangle contains two points, it leaves us no choice as to the edges we can add. It is
easy to see that the resulting support contains $K_{r,r}$ as a topological minor. Further, even for this restricted problem,
the analysis of our algorithm is highly non-trivial, and we hope that the tools introduced in this paper will 
be of wider interest.

Raman and Ray \cite{RR18},
showed that the hypergraph defined by \emph{non-piercing regions}\,\footnote{A family of simply connected regions $\mathcal{R}$, each of whose boundary is defined by a simple Jordan curve is called
non-piercing if for every pair of regions $A,B\in\mathcal{R}$, $A\setminus B$ and $B\setminus A$ are connected. The result of \cite{RR18} was
for more general families.} in the plane admits a planar support. Their proof implies an $O(m^2(\min\{m^3,mn\}+n))$ time 
algorithm to compute a planar support where $m$ is the number of regions and $n$ is the number of points in the arrangement of the regions. 
While their algorithm produces a plane embedding, the edges may in general be arbitrarily complicated curves i.e., they may have an arbitrary number of bends. It can be shown that if the non-piercing regions are convex then there exists an embedding of the planar support with straight edges but it is not clear how to find such an embedding efficiently.

We present a simple and fast algorithm for drawing plane supports with straight-line edges for non-piercing rectangles. More precisely, the following is the
problem definition:

\smallskip
\noindent
{\bf Support for non-piercing rectangles:}\\
{\bf Input:} A set of $m$ axis-parallel non-piercing rectangles $\mathcal{R}$ and a set $P$ of $n$ points in $\mathbb{R}^2$. \\
{\bf Output:} A plane graph $G$ on $P$ s.t. for each $R\in\mathcal{R}$,
$G[R\cap P]$, namely the induced subgraph on the points in $R\cap P$, is connected.

Our algorithm runs in $O(n \log^2 n + (n + m) \log m)$ time, and can be easily implemented using existing data structures. The embedding computed by our algorithm not only has straight-line edges but also 
for each edge $e$, the axis-parallel rectangle with $e$ as the diagonal does not contain any other point of $P$ --
this makes the visualization cleaner.

In order to develop a faster algorithm, we need to find a new construction (different from \cite{RR18}), and the proof of correctness for this construction is not so straightforward. We use a sweep line algorithm. However, at any point in time, it is not possible to have the invariant that the current graph is a support for the portions of the rectangles that lie to the left of the sweep line. Instead, we show that certain {\em slabs} within each rectangle induce connected components of the graph and only after we sweep over a rectangle completely do we finally have the property that the set of points in that rectangle induce a connected subgraph. 

{\bf Organization.} The rest of the paper is organized as follows. We start in Section~\ref{sec:related} with related work.
In Section~\ref{sec:prelim}, we present preliminary notions required for our algorithm.
In Section~\ref{sec:algorithm}, we present a fast algorithm to construct a planar support.
We show in Section~\ref{sec:correctness} that
the algorithm is correct, i.e., it does compute a planar support. We present the implementation details in Section~\ref{sec:implementation}. 
We give some applications of our results in Section~\ref{sec:Application}, and conclude in
Section \ref{sec:conclusion}.

\section{Related work}
\label{sec:related}

The notion of the existence of a support, and in particular a planar support arose in the field of
VLSI design \cite{voloshina1984planarity}. A VLSI circuit is viewed as a hypergraph
where each individual electric component corresponds uniquely to a vertex of the hypergraph,
and sets of components called \emph{nets} correspond uniquely to a hyperedge.
The problem is to connect the components with wires so that
for every net, there is a tree spanning its components. Note that planarity in this context is natural as 
we don't want wires to cross. 

Thus, a motivation to study supports was to define a notion of planarity suitable for hypergraphs. 
Unlike for graphs, there are different notions of planarity of hypergraphs, not all equivalent to each other. 
Zykov \cite{zykov1974hypergraphs} defined a notion of planarity that was more restricted. A hypergraph is said to be
Zykov-planar if its incidence bipartite graph is planar \cite{zykov1974hypergraphs,walsh1975hypermaps}. 

Johnson and Pollack \cite{johnson1987hypergraph} showed that deciding if a hypergraph admits a planar support is NP-hard. 
The NP-hardness result was sharpened by Buchin, et al., \cite{buchin2011planar} who showed that deciding if a hypergraph
admits a support that is a $k$-outerplanar graph, for $k\ge 2$ is NP-hard, and showed that we can
decide in polynomial time if a hypergraph admits a support that is a tree of bounded degree. 
Brandes, et al., \cite{brandes2010blocks} showed that we can decide in polynomial time if a hypergraph admits a support
that is a cactus\footnote{A cactus is a graph where each edge of the graph lies in at most one cycle.}.

Brandes et al., \cite{brandes2012path}, motivated by the drawing of metro maps, considered the problem of constructing \emph{path-based
supports}, which must satisfy an additional property that the induced subgraph on each hyperedge contains a Hamiltonian path on the vertices of the hyperedge.

Another line of work, motivated by the analysis of approximation algorithms for packing and covering problems
on geometric hypergraphs started with the work of Chan and Har-Peled \cite{ChanH12}, and Mustafa and Ray \cite{mustafa2010improved}.
The authors showed, respectively, that for the Maximum Packing\footnote{In a Maximum Packing problem,
the goal is to select the largest subset of pairwise disjoint hyperedges of a hypergraph.}
of pseudodisks\footnote{A set of simple Jordan curves is a set of pseudocircles if the curves pairwise intersect twice or zero times.
The pseudocircles along with
the bounded region defined by the curves is a collection of pseudodisks.}, and for the
Hitting Set\footnote{In the Hitting Set problem, the goal is to select the smallest subset of vertices of a hypergraph so that
each hyperedge contains at least one vertex in the chosen subset.} problem for pseudodisks, a simple \emph{local search} algorithm
yields a PTAS.
These results were extended by Basu Roy, et al., \cite{BasuRoy2018} 
to work for the Set Cover and Dominating Set\footnote{In the Set Cover problem, the input
is a set system $(X,\mathcal{S})$ and the goal is to select the smallest sub-collection $\mathcal{S}'$ that covers the elements in $X$. For a graph,
a subset of vertices $S$ is a dominating set if each vertex in the graph
is either in $S$, or is adjacent to a vertex in $S$.} problems defined by
points and non-piercing regions, 
and by Raman and Ray \cite{RR18}, who gave a general theorem on the existence of a planar support
for any geometric hypergraph defined by two families 
of \emph{non-piercing regions}.
This result generalized and unified the previously mentioned results, and for a
set of $m$ non-piercing regions, and a set of $n$ points
in the plane, it implies that a support graph
can be constructed in 
time $O(m^2(\min\{m^3,mn\}+n))$. It follows that for non-piercing axis-parallel rectangles, a planar
support can be constructed in time $O(m^2(\min\{m^3,mn\}+n))$. However, in the embedding of the support thus constructed, the edges may be drawn as arbitrary
curves.

\section{Preliminaries}
\label{sec:prelim}
Let $\mathcal{R} = \{R_1,\ldots, R_m\}$ denote a set of axis-parallel rectangles 
and let $P=\{p_1,\ldots, p_n\}$ denote a set of points in the plane. 
We assume that the rectangles and points are in \emph{general position}, i.e., the points in $P$
have distinct $x$ and $y$ coordinates, and the boundaries of any two rectangles in $\mathcal{R}$ are defined by distinct $x$-coordinates and 
distinct $y$-coordinates. Further, we assume that no point in $P$ lies on
the boundary of a rectangle in $\mathcal{R}$.

\textbf{Piercing, Discrete Piercing. }
A rectangle $R'$ is said to \emph{pierce} a rectangle $R$ if $R\setminus R'$
consists of two connected components.
A collection $\mathcal{R}$ of rectangles is \emph{non-piercing} if
no pair of rectangles pierce.
A rectangle $R'$ \emph{discretely pierces} a rectangle $R$ if 
$R'$ pierces $R$ and each component of $R\setminus R'$ contains a point of $P$.
Since we are primarily concerned with discrete piercing, the phrase ``$R$ pierces $R'$''
will henceforth mean discrete piercing, unless stated otherwise.
Note that while piercing is a symmetric relation, discrete piercing is not.

\textbf{`L'-shaped edge. }
We construct a drawing of a support graph $G$ on $P$
using `L'-shaped edges of type: 
$\upleft[0.2]$ or $\leftup[0.2]$. Henceforth, the term \emph{edge} will mean one of the two
`L'-shaped edges joining two points. The embedded graph may not be planar due to the overlap of the edges along their horizontal/vertices segments. However, as we show, $G$ satisfies the additional property that for each edge, the axis-parallel rectangle defined by the
edge has no points of $P$ in its interior (formal definition below), and that no pair of edges cross. Consequently,
replacing each edge with the straight segment joining its end-points
yields a plane embedding of $G$.

\textbf{Delaunay edge, Valid edge, $\mathbf{R(\cdot), h(\cdot), v(\cdot)}$. }
For an edge between 
points $p,q\in P$, let $R(pq)$ denote the rectangle
with diagonally opposite corners $p$ and $q$. 
The edge $pq$ is a
\emph{Delaunay edge} if the interior of $R(pq)$ does not contain a point of $P$.
We say that an edge $pq$ (discretely) pierces a rectangle $R$ if $R\setminus\{pq\}$ consists of two regions, and each region contains a point
of $P$.
An edge $pq$ is said to be \emph{valid} if it does not discretely pierce any rectangle $R \in \mathcal{R}$, and does not cross any existing
edge. 

For an edge $pq$, we use $h(pq)$ for the horizontal segment of $pq$, and $v(pq)$ for the vertical segment of $pq$. 

\textbf{Monotone Path, Point above Path. }
A path $\pi$ is said to be $x$-monotone if a vertical line, i.e., a line 
parallel to the $y$-axis,
does not intersect the path in more than one point. We modify this
definition slightly for our purposes - we say that a path consisting of a sequence of
$\upleft[0.2]$, or $\leftup[0.2]$ edges is $x$-monotone if any vertical line intersects the path in at most one vertical segment (which may in some cases be a single point). 
Let $\pi$ be a path and
$q$ be a point not on the path. We say that ``$q$ lies above $\pi$'' if $\ell_q$,
the vertical line through $q$ intersects $\pi$ at point(s) below $q$. We define the notion that ``$q$ lies below $\pi$'' analogously. Note that these notions are defined only if $\ell_q$ intersects $\pi$.

\textbf{Left(Right)-Neighbor, Left(Right)-Adjacent.}
For a point $q\in P$ and a set $P'\subseteq P$, the \emph{right-neighbor}
of $q$ in $P'$ is $q_1$, where $q_1=\argmin_{q'\in P'} \{ x(q'): x(q')>x(q)\}$.
The \emph{left-neighbor} of $q$ in $P'$ is defined similarly, i.e.,
$q_0$ is the left-neighbor of $q$, where $q_0=\argmax_{q'\in P'}\{x(q'): x(q')<x(q)\}$.
 Note that being a left- or right-neighbor
is a \emph{geometric notion}, and not related to the support graph we construct.
We use the term \emph{left-adjacent} to refer to the neighbors of
$q$ in a plane graph $G$ that lie to the left of $q$. The term \emph{right-adjacent} is defined analogously.

\section{Algorithm}
\label{sec:algorithm}
In this section, we present an algorithm to compute a planar support for the hypergraph defined by points and non-piercing axis-parallel rectangles in $\mathbb{R}^2$:
Perform a left-to-right vertical line sweep and at each input point encountered, add all possible \emph{valid Delaunay edges} to previous points.
The algorithm, presented as Algorithm~\ref{alg:prim}, 
draws edges having shapes in \{\upleft[0.2], \leftup[0.2]\}. 
We prove correctness of Algorithm~\ref{alg:prim} in Section~\ref{sec:correctness}, and show how it can be implemented to run in $O(n\log^2 n + (n + m) \log m))$ time in
Section~\ref{sec:implementation}.

\begin{algorithm}
\SetAlgoLined
\KwIn{ A set $P$ of points, and a set $\mathcal{R}$ of non-piercing axis-parallel rectangles in $\mathbb{R}^2$.}
\KwOut{Embedded Planar Support $G = (P, E)$}
 Order $P$ in increasing order of \textit{x-coordinates}: $(p_1,\ldots, p_n)$ \\
 $E = \emptyset$ \\
 \For{{\rm each point} $p_i$ {\rm in sorted order}, $i \in \{2, 3, \ldots, n\}$}{
    $E = E \cup \{e_{ij}=p_ip_j \,|\, j<i,\mbox{ and } e_{ij}$ \mbox{ is a} \emph{valid Delaunay edge}.\} \\
 }
\caption{The algorithm outputs a graph $G$ on $P$ embedded in $\mathbb{R}^2$, whose edges are valid Delaunay edges of type $\{\upleft[0.2],\leftup[0.2]\}$. Replacing each Delaunay edge $\{p,q\}$ by the diagonal of $R(pq)$ yields a plane embedding of $G$.}
\label{alg:prim}
\end{algorithm}

\subsection{Correctness}
\label{sec:correctness}
In this section, we show that the graph $G$ constructed on $P$ by Algorithm~\ref{alg:prim} 
is a support graph for the rectangles in $\mathcal{R}$, and this is sufficient as planarity follows directly by construction.
The proof is technical, and we start with some necessary notation.

For a rectangle $R$, we denote the $y$-coordinates of the lower and upper horizontal sides by $y_-(R)$ and $y_+(R)$, respectively. Similarly, 
$x_-(R)$ and $x_+(R)$ denote respectively the $x$-coordinates of the left and right vertical sides. 
We denote the vertical line 
through any point $p$ by $\ell_p$. 

We use $\textsc{Piece}(R,H)$ to denote the rectangle $R\cap H$ for a halfplane $H$ defined by a vertical line.
We abuse notation and use $\textsc{Piece}(R,p)$ to denote the rectangle $R\cap H_-(\ell_p)$, the intersection of $R$ with the left half-space
defined by the vertical line through the point $p$.

We also use the notation $R[x_-,x_+]$ to denote the sub-rectangle of
rectangle $R$, that lies between $x$-coordinates $x_-$ and $x_+$. Similarly, 
we use $R[y_-,y_+]$ to denote the sub-rectangle of $R$ that lies between the $y$-coordinates $y_-$ and $y_+$. 

To avoid boundary conditions in the definitions that follow, we add two rectangles: $R_{top}$ above all rectangles in $\mathcal{R}$, and $R_{bot}$ below
all rectangles in $\mathcal{R}$, that is $y_-(R_{top}) > \max_{R\in \mathcal{R}} y_+(R)$, and $y_+(R_{bot}) <\min_{R\in \mathcal{R}} y_-(R)$.
The rectangles $R_{top}$, and $R_{bot}$ span the width of all rectangles, i.e., $x_-(R_{top}) = x_-(R_{bot}) < \min_{R\in\mathcal{R}} x_-(R)$, and
$x_+(R_{top}) = x_+(R_{bot}) > \max_{R\in\mathcal{R}} x_+(R)$.
We add two points $P_{top} = \{p^+_1, p^+_2\}$ to the interior of $R_{top}$, and two points $P_{bot} = \{p^-_1, p^-_2\}$  to the interior of $R_{bot}$,
such that $x(p^+_1) = x(p^-_1) < \min_{p\in P} x(p)$, and $x(p^+_2) = x(p^-_2) > \max_{p\in P} x(p)$.
Let $\mathcal{R}' = \mathcal{R}\cup \{R_{top},R_{bot}\}$, and $P' = P\cup P_{top} \cup P_{bot}$.
For ease of notation, we simply use $\mathcal{R}$ and $P$ to denote $\mathcal{R}'$ and $P'$ respectively, and implicitly assume the existence of $R_{top},R_{bot},
P_{top}$ and $P_{bot}$.

For a vertical segment $s$, a rectangle $R\in\mathcal{R}'$ is said to be \emph{active} at $s$ if it is either
discretely
pierced by $s$ i.e., $R\setminus s$ is not connected and each of the two components contains a point of $P$, or there is a point of $P\cap s$ in $R$.
We denote the set of all active rectangles at $s$ by $\Active(s)$.
For a point $p\in P\cap s$, we define $\Contain(s,p)$ to be
the set of rectangles in $\Active(s)$ that contains the point $p$.
We define $\Above(s,p)$ to be the set
of rectangles in $\Active(s)$ that lie strictly above $p$, i.e., $\Above(s,p) = \{R\in\Active(s): y_-(R)>y(p)\}$.
Similarly, $\Below(s,p)=\{R\in\Active(s):y_+(R)<y(p)\}$. It follows that for any point $p\in s$,
$\Active(s)=\Contain(s,p)\sqcup\Above(s,p)\sqcup\Below(s,p)$, where $\sqcup$ denotes disjoint union. 

Note that for the vertical line $\ell_p$ through $p\in P$,  $\Active(\ell_p)\neq\emptyset$, as $\Active(\ell_p)$
contains the rectangles $R_{top}$ and
$R_{bot}$. Similarly, $\Above(\ell_p,p)\neq\emptyset$ and $\Below(\ell_p,p)\neq\emptyset$.
Abusing notations slightly, we write $\Active(p)$ instead of 
$\Active(\ell_p)$, and likewise with 
$\Contain(\cdot), \Above(\cdot)$ and 
$\Below(\cdot)$.

For a point $p\in P$, we now introduce the notion of \emph{barriers}. Any active rectangle $R'$ in $\Above(p)$ prevents a valid Delaunay edge incident
on $p$ from being incident to a point to the left of $p$ above $y_+(R')$, as such an edge would discretely pierce $R'$. 
Hence, among all rectangles $R'\in \Above(p)$, the one with lowest
$y_+(R')$ is called the \emph{upper barrier} at $p$, denoted $\ub(p)$. Thus, 
$
\ub(p)=\argmin_{\substack{R'\in\Above(p)}}y_+(R')
$.

Similarly, we define the \emph{lower barrier} of $p$, 
$
\lb(p)=\argmax_{\substack{R'\in\Below(p)}}y_-(R')
$.

Note that $\ub(p)$ and $\lb(p)$ exist for any $p\in P$ since 
 $\Above(p)$ and $\Below(p)$ are non-empty.
 
While the rectangles in $\mathcal{R}'$ are non-piercing, a rectangle $R'\in\Active(p)$ can be discretely pierced by $\Piece(R,p)$.
We thus define the \emph{upper piercing barrier} $\upb(R,p)$ as the rectangle $R'\in\Above(p)$ with the lowest $y_+(R')$
that is pierced by $\Piece(R,p)$, and we define the \emph{lower piercing barrier} $\lpb(R,p)$ analogously. That is,

\begin{center}
\begin{tabular}{lcr}
$\begin{aligned}
\upb(R,p) = \argmin_{\substack{R'\in\Above(p)\\ \Piece(R,p)\mbox{ pierces } R'}} y_+(R')
\end{aligned} $
& and &
$\begin{aligned}
\lpb(R,p) = \argmax_{\substack{R'\in\Below(p)\\ \Piece(R,p)\mbox{ pierces } R'}} y_-(R')
\end{aligned}$
\end{tabular}
\end{center}

For a point $p\in P$ and a rectangle $R\in\Contain(p)$, if $\upb(R,p)$ or $\lpb(R,p)$ exist,
then the horizontal line containing $y_+(\upb(R,p))$ together with the horizontal line containing
$y_-(\lpb(R,p))$ naturally split $\Piece(R,p)$ into at most three sub-rectangles called \emph{slabs}.
The point $p$ lies in exactly one of these slabs,
denoted $\Slab(R,p)$. Thus, $\Slab(R,p)$ is the sub-rectangle of $R$ whose left and right-vertical sides are respectively defined
by $x_-(R)$ and $\ell_p$, and
the upper and lower sides are respectively defined by

\begin{figure}[t!]
\centering
\vspace{-2em}
\includegraphics[width=0.5\textwidth]{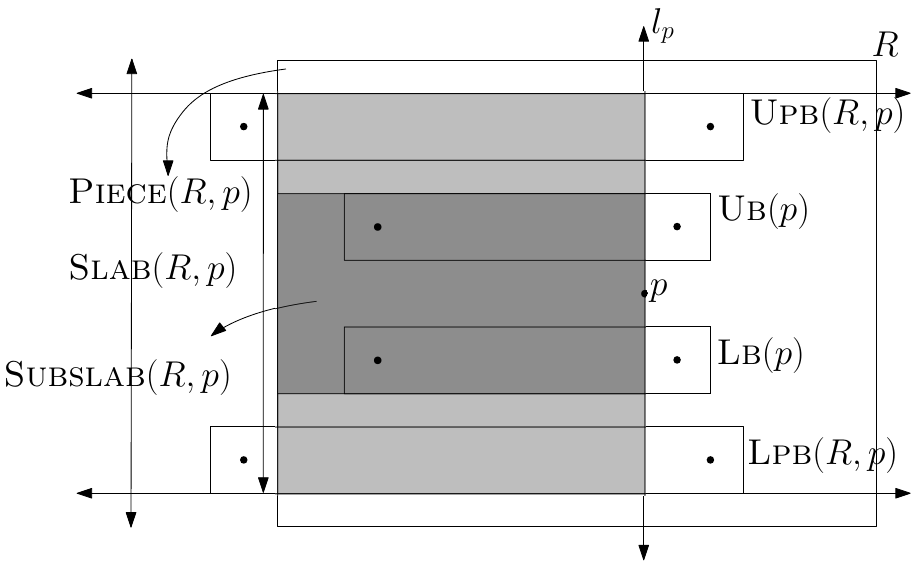}
\caption{The figure above shows $\ub(p),\lb(p)$, and the upper and lower piercing barriers $\lpb(R,p)$ and $\upb(R,p)$ of $\Piece(R,p)$.
The slab $\Slab(R,p)$ containing
$p$ defined by $\upb(R,p)$ and $\lpb(R,p)$ is shaded.
The dark grey part shows the $\Subslab(R,p)$.}
\label{fig:linebarrierslicer}
\end{figure}

\[
y_+(\Slab(R,p)) = \left\{\begin{array}{ll} y_+(\upb(R,p)), & \mbox{if } \upb(R,p)\mbox{ exists}\\
                                           y_+(R), & \mbox{ otherwise}\end{array}\right.
\]

and similarly,
\[
y_-(\Slab(R,p)) = \left\{\begin{array}{ll} y_-(\lpb(R,p)), & \mbox{if } \lpb(R,p)\mbox{ exists}\\
                                           y_-(R), & \mbox{ otherwise}\end{array}\right.
\]

By definition, for a point $p$ and $R\in\Contain(p)$, if $\upb(R,p)$ exists, then $y_+(\upb(R,p)) \allowbreak \ge \allowbreak y_+(\ub(p))$.
Similarly, if $\lpb(R,p)$ exists, then ${y_-(\lpb(R,p))\le y_-(\lb(p))}$.
Thus, $y_+(\ub(p))$ and $y_-(\lb(p))$ together
split $\Slab(R,p)$ further into at most 3 sub-rectangles called \emph{sub-slabs} whose vertical sides coincide with the vertical sides of $\Slab(R,p)$, and the horizontal sides are defined by $y_+(\ub(p))$ and $y_-(\lb(p))$.
Let $\Subslab(R,p)$ denote the sub-slab containing $p$. 
Figure~\ref{fig:linebarrierslicer} illustrates the notions defined thus far. 
Note that the left-adjacent vertices of $p$ in $G$ that are contained in $R$, only lie in $\Subslab(R,p)$.

\smallskip
\noindent
{\bf Proof Strategy:} 
To prove that the graph $G$ constructed by Algorithm~\ref{alg:prim} is a support for $\mathcal{R}$, we proceed in two steps.
First (and the part that requires most of the work) we show that for each $R\in\mathcal{R}$ and $p\in P\cap R$, the subgraph of $G$ induced by the points in $\Slab(R,p)$
is connected. Second, we show that if $p$ is the rightmost point in $R$, then $\Slab(R,p)$ contains all points in $R\cap P$
which, by the first part, is connected.

When processing a point $p$, Algorithm~\ref{alg:prim} only
adds valid Delaunay edges from $p$ to points to its left. That is, we only add edges to a subset of points in 
$\Subslab(R,p)$. To show that $\Slab(R,p)$ is connected, one approach could be to show that the $\Slab(R,p)$ is
covered by sub-slabs defined by points in $\Slab(R,p)$, adjacent sub-slabs share a point of $P$, and that
points in a sub-slab induce a connected subgraph. Unfortunately, this is not true, and we require a finer partition
of a slab.
We proceed as follows: 
First, we define a sequence of sub-rectangles of $\Slab(R,p)$ called \emph{strips}, 
denoted $\Strip(R,p,i)$ for $i\in\{-t,\ldots, k\}$, where the strips that lie above $p$ have positive indices, 
the strips that lie below $p$ have negative indices, and the unique strip that contains $p$ has index $0$. 
Further, each strip shares its vertical sides with $\Slab(R,p)$. 
In the following, since $R$ and $p$ are fixed, we refer to 
$\Strip(R,p,i)$ as $\Strip_i$. We define the strips
so that they satisfy the following conditions:
\begin{enumerate}[(s-i)]
\item Each strip is contained in the slab, i.e, $\Strip_i\subseteq\Slab(R,p)$ for each $i\in\{-t,\ldots, k\}$\label{cond:S1}.
\item The union of strips cover the slab, i.e., $\Slab(R,p)\subseteq \cup_{i=-t}^{k}\Strip_i$, and\label{cond:S2}
    \item Consecutive strips contain a point of $P$ in their intersection, i.e., $\Strip_i\cap\Strip_{i-1}\cap P \neq\emptyset$ for all $i\in\{-t+1,\ldots, k\}$. Consequently, each strip contains a point of $P$. \label{cond:S4}
\end{enumerate}

In order to prove that $\Slab(R,p)$ is connected, we describe below a strategy that does not quite work 
but, as we show later, can be fixed.
 
Let $\Strip_i\cap P = P_i$. By Condition (s-\ref{cond:S4}), $P_i\neq\emptyset$ for any $i\in\{-t,\ldots, k\}$.
For a strip $\Strip_i$, let $p_i$ denote the rightmost point in it. 
Let us assume for now that for each $i\in\{-t,\ldots, k\}$, and 
each point $q\in P_i$, there is a path from $q$ to $p_i$ that lies entirely in $\Strip_i$.\footnote{This assumption is incorrect but will be remedied later.}
Now, consider an arbitrary point $q\in\Slab(R,p)$.
By Condition (s-\ref{cond:S2}) each point in $P\cap \Slab(R,p)$ is contained in at least one strip. Therefore, $q\in\Strip_i$ for some $i\in\{-t,\ldots, k\}$.
By our assumption, there is a path $\pi^1_i$ from $q$ to $p_i$ that lies entirely in $\Strip_i$.
If $i \ge 0$ (a symmetric argument works when $i < 0$), since Condition (s-\ref{cond:S4}) implies
consecutive strips intersect at a point in $P$, there is a path $\pi^2_i$ from $p_i$ to a point $q'\in P_i\cap P_{i-1}$ that lies entirely in
$\Strip_i$.
Again, by our assumption, there is a path $\pi^1_{i-1}$ from $q'$ to $p_{i-1}$ that lies entirely in $\Strip_{i-1}$. 
Repeating the argument above with $i-1, i-2,\ldots$, until $i=0$, and concatenating the paths $\pi^1_i, \pi^2_i,\pi^1_{i-1},\ldots$, we obtain a path $\pi$ from $q$ to $p$, each sub-path of which is a path from a point in a strip to the rightmost point in that strip such that each point in the path lies entirely in the strip.
By Condition (s-\ref{cond:S1}), $\Strip_i\subseteq\Slab(R,p)$ for each $i\in\{-t,\ldots, k\}$. Therefore, $\pi$ lies entirely in $\Slab(R,p)$. 
Since $q$ was arbitrary, this implies that $\Slab(R,p)$ is connected. 

Consider a slab $\Slab(R,p)$ corresponding to a rectangle $R\in\mathcal{R}$ 
and a point $p\in P\cap R$. The strips corresponding to $\Slab(R,p)$ are defined 
as follows:
Let $s$ denote the open segment of $\ell_p$ 
between $p$ and $y_+(\Slab(R,p))$ of $\ell_p$, 
the vertical line through $p$. Let $\mathcal{R}_s=(R_0, \ldots, R_{h})$ 
be the rectangles in
$\Active(s)$ ordered by their upper sides i.e., $ y_+(R_i)<y_+(R_j)$, for
$0\le i < j\le h$.
Similarly, let $s'$ denote the open segment of $\ell_p$ between $p$ and
$y_-(\Slab(R,p))$ and  let $\mathcal{R}_{s'}=(R'_0,\ldots, R'_{h'})$ denote the 
rectangles in $\Active(s')$ ordered by their lower sides $y_-(R'_i)>y_-(R'_j)$ for $0\le i < j\le h'$. 

We define  $\Strip_0=\Slab(R,p)[y_-(R'_0), y_+(R_0)]$, if $\Active(s)\neq\emptyset$ and
$\Active(s')\neq\emptyset$. If $\Active(s)=\emptyset$, we set $y_+(\Strip_0)=y_+(\Slab(R,p))$. Similarly,
if $\Active(s')=\emptyset$, we set $y_-(\Strip_0)=y_-(\Slab(R,p))$. We set $p_0=p$.
Having defined $\Strip_0$, we set $\mathcal{R}_s=\mathcal{R}_s\setminus R_0$ and $\mathcal{R}_{s'}=\mathcal{R}_{s'}\setminus R'_0$.

\begin{figure}[h!]
\begin{center}
\includegraphics[width=2in]{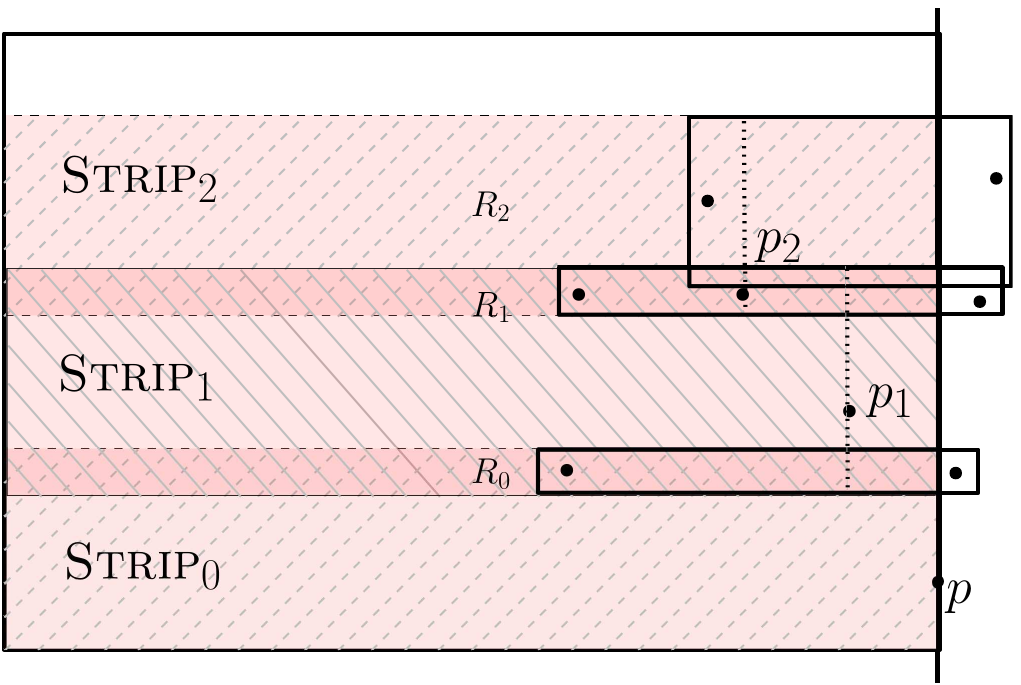}
\caption{The figure shows the construction of the strips $\Strip_0$, $\Strip_1$ and $\Strip_2$. The vertical
line segment through $p_i$, $i\in\{1,2\}$ shows that $p_i$ is the rightmost point among the points in the strip $i$.}
\label{fig:stripConstruction}
\end{center}
\end{figure}

For $i > 0$, having constructed $\Strip_j$ for $j=0,\ldots, i-1$,
we do the following while $\mathcal{R}_s\neq\emptyset$:
Let $S_{i}=\argmin_{R'\in\mathcal{R}_s} y_-(R')$, and let
$y_-=y_-(S_{i})$.
Let $R_i = \argmin\{y_+(R'): R'\in\mathcal{R}_s: y_-(R')>y_-\}$, 
and let $y_+=\min\{y_+(\Slab(R,p)), y_+(R_i)\}$.
Set $y_-(\Strip_i)=y_-$ and $y_+(\Strip_i)=y_+$.
Let $p_i=\argmax\{x(p'):p'\in P\cap\Slab(R,p): y_-<y(p')<y_+\}$.
Note that $p_i$ exists since $S_{i}\in\Active(s)$.
Set $\mathcal{R}_s=\mathcal{R}_s\setminus\{R':y_-(R')<y_-(R_i)\}$.

For $i < 0$, the construction is symmetric. Having constructed
$\Strip(R,p,j)$ for $j$ from $0$ down to $-i+1$, we do the following until
$\mathcal{R}_{s'}=\emptyset$. Let
$S'_i=\argmax_{R'\in\mathcal{R}_{s'}} y_+(R')$, and let $y_+=y_+(S'_i)$.
Let $R'_i=\argmax\{y_-(R'): R'\in\mathcal{R}_{s'}, y_+(R')<y_+\}$.
Let $y_-=\max\{y_-(R'_i),y_-(\Slab(R,p))\}$.
Set $y_-(\Strip_i)=y_-$ and $y_+(\Strip_i)=y_+$.
Let $p_i=\argmax\{x(p'):
p'\in P\cap\Slab(R,p), y_-<y(p')<y_+\}$.
Again, $p_i$ exists since
$S'_i\in\Active(s')$.
Set $\mathcal{R}_{s'}=\mathcal{R}_{s'}\setminus\{R'\in \mathcal{R}_{s'}: y_+(R')>y_+(R'_i)\}$.
Figure \ref{fig:stripConstruction} 
illustrates the construction of the strips.

\begin{restatable}{proposition}{propstrip}\label{prop:strip}
For $i\in\{-t,\ldots, k\}$,
\begin{align*}
y_-(\lb(p_i))\le y_-(\Strip_i)<y_+(\Strip_i) \le y_+(\ub(p_i))
\end{align*}
\end{restatable}
\begin{proof}
Fix $i\in\{-t,\ldots, k\}$ and assume $i\ge 0$. For $i<0$, the proof is symmetric. 
Since $p_i\in\Strip_i$ and by the
definition of the lower barrier, $y_+(\lb(p_i))<y(p_i)<y_+(\Strip_i)$. 
If $y_-(\lb(p_i))>y_-(\Strip_i)$, since $\lb(p_i)\in\mathcal{R}_s$ and
$y_+(\Strip_i)\le\min\{y_+(R')\in\mathcal{R}_s: R'\in\mathcal{R}_s\mbox{ and } y_-(R')>y_-(\Strip_i)\}$, it implies
$y_+(\Strip_i)\le y_+(\lb(p_i))$, contradicting $p_i\in\Strip_i$. 

Now we argue about $\ub(p_i)$.
If $y_+(\ub(p_i))>y_+(\Slab(R,p))$, since $y_+(\Slab(R,p))\ge y_+(\Strip_i)$,
we have $y_+(\ub(p_i))\ge y_+(\Strip_i)$.
Otherwise, we have $\ub(p_i)\in\mathcal{R}_s$.
Since $p_i\in\Strip_i$ and 
by definition of the upper barrier, we have $y_-(\ub(p_i))>y(p_i)>y_-(\Strip_i)$.
Since
$y_+(\Strip_i)\le\min\{y_+(R'): R'\in\mathcal{R}_s\mbox{ and } y_-(R')>y_-(\Strip_i)\}$,
it follows that
$y_+(\Strip_i)\le y_+(\ub(p_i))$.
\end{proof}

\begin{restatable}{lemma}{strips}\label{lem:strips}
The strips constructed as above satisfy the following conditions:
$(i)$ $\Strip_i\subseteq\Slab(R,p)$ for each $i\in\{-t,\ldots, k\}$\label{cond:s1}.
$(ii)$ $\Slab(R,p) = \cup_{i=-t}^{k}\Strip_i$, and\label{cond:s2}
$(iii)$ $\Strip_i\cap\Strip_{i-1}\cap P \neq\emptyset$ for all $i\in\{-t+1,\ldots, k\}$. \label{cond:s4}
\end{restatable}
\begin{proof}
Item $(i)$ and $(ii)$ follow directly by construction.
For $(iii)$, note that adjacent strips contain a piece of an active rectangle and hence their intersection contains a point of $P$. 
\end{proof}

Unfortunately, our assumption that every point in a strip
has a path to its rightmost point in the strip is not correct.
To see this, consider a strip that is pierced by a rectangle $R'$, whose intersection with the strip does not contain a point of $P$. Therefore, a point in the strip that lies to the left of $R'$ cannot have a path to $p_i$ that
lies in the strip unless some of its edge is allowed to pierce $R'$.
In order to remedy this situation, we introduce the notion of a \emph{corridor}. A corridor
corresponding to a strip is a region of $\Slab(R,p)$ that contains all points
in the strip, and such that each point in the strip has a path to the rightmost point in it that lies entirely in the corridor. 
Since each corridor lies in $\Slab(R,p)$, 
the proof strategy can be suitably modified to show that $\Slab(R,p)$ is connected. 


We now define the corridors associated with each strip. 
Recall that $G=(P,E)$ is the graph constructed by Algorithm \ref{alg:prim}.
For a point $q\in P$, recall that the neighbors of $q$ in $G$ that lie
to its left are called its \emph{left-adjacent} points. If $q$ lies on a path $\pi$,
and $q'$ is the left-adjacent point of $q$ on $\pi$, then we say that $q'$ is
left-adjacent to $q$ on $\pi$.
We start with the following
proposition that will be useful in constructing the corridors.

\begin{restatable}{proposition}{deledges}\label{prop:deledges}
For a point $q\in P$,
if $(q_1,\ldots, q_r, q_{r+1},\ldots, q_s)$ is the sequence of
its left-adjacent points in $G$ 
s.t. $y(q_1)>\ldots> y(q_r)>y(q)$ and $y(q_{r+1})<\ldots< y(q_s)<y(q)$.  
Then, for $1\le i<j\le r$, $x(q_i)>x(q_j)$, and for $r+1\le i<j\le s$, $x(q_i)<x(q_j)$.
\end{restatable}
\begin{proof}
This follows directly from the fact that each edge in $G$ is Delaunay.
\end{proof}

For a strip $\Strip_i$, we define its corresponding corridor $\corridor_i$ as follows: 
The corridor is the region of $\Slab(R,p)$ bounded
by two paths: an \emph{upper path} $\pi^u_i$, and a \emph{lower path} $\pi^\ell_i$,
defined as follows.

The upper path $\pi^u_i = (q_0, q_1, q_2, \ldots,)$ is constructed by starting with $j = 0, q_0 = p_i$, and repeating (1) set $q_{j + 1} \gets q'$ where $q'$ has the highest $y(q')$ among the left-adjacent points of $q_j$. 
(2) $j \gets j + 1$. 
We stop when we cannot find such a $q'$ for the current $q_j$ in $G$ that lies in $\Slab(R,p)$, 
where we complete the path by following the edge to the left-adjacent vertex $q_{j+1}$ of $q_j$
with highest $y$-coordinate. Thus, the last vertex of $\pi^u_i$ possibly does not lie in $\Slab(R,p)$.

The lower path $\pi^{\ell}_i=(q'_0,\ldots,)$ is constructed similarly. 
For $j=0$, set $q'_0=p_i$, and repeating $(1)$ set $q'_{j+1}\gets q'$,
where $q'$ has the \emph{lowest} $y(q')$ among the left-adjacent points of $q'_j$.
We stop when we cannot find such a $q'$ in $\Slab(R,p)$ for the current $q'_j$, where we complete
the path by following the edge from $q'_j$ its left-adjacent vertex $q'_{j+1}$ with
smallest $y$-coordinate. Thus, the last vertex of $\pi^\ell_i$ possibly does not lie in $\Slab(R,p)$.

$\corridor_i$ is
the region of $\Slab(R,p)$ that lies between the upper and lower paths $\pi^{u}_i$ and $\pi^{\ell}_i$. Figure \ref{fig:corridor} 
shows a corridor corresponding to a strip.
We start with some basic observations about the corridors thus constructed.

\begin{figure}[ht!]
\begin{center}
\includegraphics[width=2in]{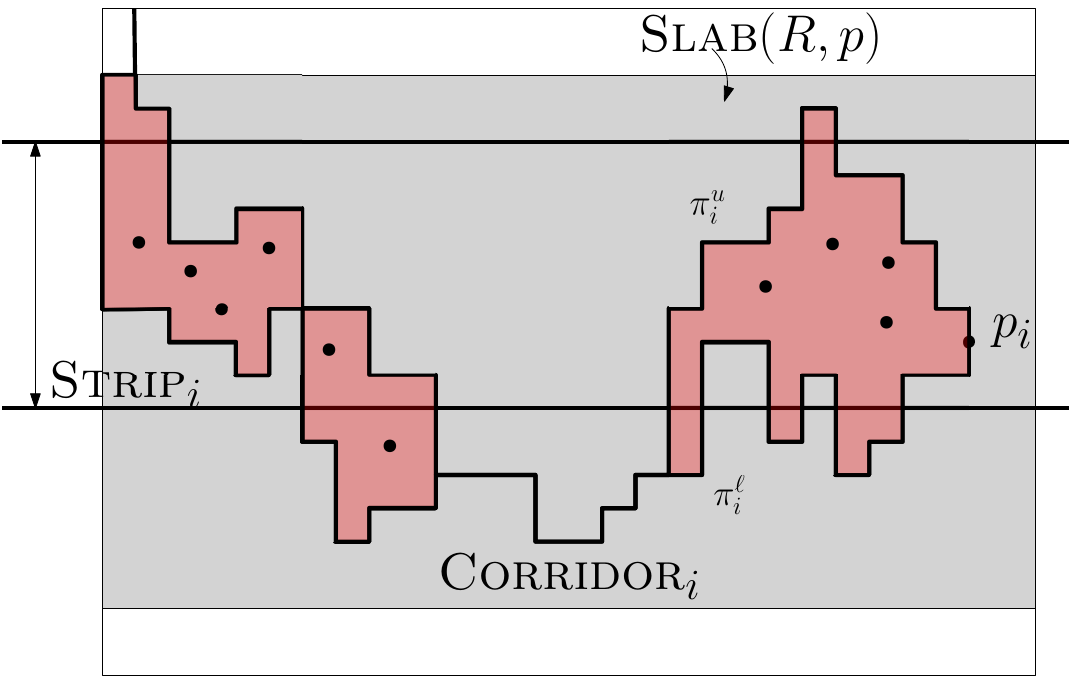}
\caption{The figure above shows the strip $\Strip_i$, the slab $\Slab(R,p)$ in grey, and the corridor $\corridor_i$ as the region shaded in red between $\pi^u_i$ and $\pi^{\ell}_i$.}
\label{fig:corridor}
\end{center}
\end{figure}

\begin{restatable}{proposition}{xmonotone}\label{prop:xmonotone}
For $i\in\{-t,\ldots, k\}$, $\pi^u_i$ and $\pi^{\ell}_i$ are $x$-monotone.
\end{restatable}
\begin{proof}
This follows directly by construction since at each step we augment the path by adding to it, the left-adjacent neighbor to the current vertex of the path.
\end{proof}

The graph $G$ constructed by Algorithm \ref{alg:prim} do not cross. 
We say that two paths $\pi_1$ and
$\pi_2$ in $G$ \emph{cross} if there is an $x$-coordinate at which $\pi_1$
lies above $\pi_2$, and an $x$-coordinate at which $\pi_2$ lies above $\pi_1$.

\begin{restatable}{proposition}{obscor}\label{obs:cor}

For $i\in\{-t,\ldots, k\}$, $\pi_i^l$ does not lie above $\pi^u_i$, and
$\pi^{\ell}_i$ and $\pi^u_i$ do not cross.
\end{restatable}
\begin{proof}
Let $\pi^u_i=(q_0, q_1,\ldots, q_r)$ and $\pi^{\ell}_i=(q'_0,\ldots, q'_s)$, 
where $q_0=q'_0=p_i$. Since $q_1$ is the left-adjacent of $p_i$ in $\Slab(R,p)$
with highest $y$-coordinate and $q'_1$ is the left-adjacent point of $p_i$
with lowest $y$-coordinate, $y(q'_1)\le y(q_1)$. Thus, $\pi^{\ell}_i$ does not
lie above $\pi^u_i$ at $x(p_i)$. If at some $x$-coordinate $x'$,
$\pi^{\ell}_i$ lies above $\pi^{u}_i$, then the paths must have crossed to the right of $x'$.

Let $q_i=q'_j=q$ be a point of $P$ common to $\pi^u_i$ and $\pi^{\ell}_i$ lying
to the left of $p_i$. Again, since 
$q_{i+1}=\arg\max\{y(q''): q''\in\Slab(R,p)\cap P, x(q'')<x(q), \{q,q''\}\in E(G) \}$,
and $q'_{j+1} = \arg\min\{y(q''): q''\in\Slab(R,p)\cap P, x(q'')<x(q), \{q,q''\}\in E(G) \}$,
it follows that $y(q'_{j+1})\le y(q_{i+1})$. Hence, the paths do not cross, and since
$\pi^{\ell}_i$ does not lie above $\pi^u_i$ at $x(p_i)$, it does not do so at any $x$-coordinate
to the left of $p_i$ either.
\end{proof}

Recall that the left-neighbor of a point $q$ in a set $P'$ is the
point $p'=\argmax_{p''\in P'} x(p'')<x(q)$. The right-neighbor is
defined similarly. Note that the left and right neighbors are defined
geometrically, and they may not be adjacent to $q$ in the graph $G$.
For a point $q\in P$ and $i\in\{-t,\ldots, k\}$, 
we let $r_0$ and $r_1$ denote respectively,
the left- and right-neighbors of $q$ on $\pi^u_i$. 
Similarly, we let $r'_0$ and
$r'_1$ denote respectively, the left- and right-neighbors of $q$ on $\pi^{\ell}_i$.

\begin{restatable}{proposition}{below}\label{prop:below}
For $i\in\{-t,\ldots, k\}$ and $q\in P$, if $q$ lies above $\pi^u_i$, then 
$y(q)>\max\{y(r_0), y(r_1)\}$. 
Similarly, if $q$ lies below $\pi^{\ell}_i$, then
$y(q)<\min\{y(r'_0),y(r'_1)\}$.
\end{restatable}
\begin{proof}
We assume $q$ lies above $\pi^u_i$. The other case 
follows by an analogous argument. 
By Proposition \ref{prop:xmonotone}, since $\pi^u_i$
is $x$-monotone, it follows that $r_0$ and $r_1$ are
consecutive along $\pi^u_i$, and thus, $r_0r_1$ is a
valid Delaunay edge in $G$. 
If either $y(r_0)>y(q)>y(r_1)$, or $y(r_0)<y(q)<y(r_1)$, then
it contradicts the fact that $r_0r_1$ is Delaunay. Hence,
$y(q)>\max\{y(r_0),y(r_1)\}$ as $q$ lies above $\pi^u_i$.
\end{proof}

We now show that the corridors constructed
satisfy the required conditions. The first condition below, follows directly from construction.

\begin{restatable}{lemma}{corridorun}\label{lem:corridor1}

For $i\in\{-t,\ldots, k\}$, 
$\corridor_i\subseteq\Slab(R,p)$. \label{cond:c2}
\end{restatable}

\begin{proof}
Follows directly by construction.
\end{proof}

Next, we show that for each strip, its corresponding corridor contains all its points,
that is all points in $P\cap \Strip_i$ are contained between the upper and lower
paths of $\corridor_i$.
Before we do that, we need the following two technical statements.

\begin{restatable}{proposition}{rightEdge}\label{prop:rightEdge}
Let $q,q'\in P$, with $x(q)<x(q')$ s.t. $qq'$ is Delaunay. If $qq'\not\in E(G)$,
then either $(i)$ $h(qq')$ pierces a rectangle,
or crosses an existing edge, or $(ii)$ $v(qq')$ pierces a rectangle. In particular, $v(qq')$ does not cross an existing edge.
\end{restatable}
\begin{proof}
The points in $P$ are processed by Algorithm \ref{alg:prim} in increasing order of 
their $x$-coordinates, and when a point is being processed, 
we add edges of type \{$\upleft[0.2]$, $\leftup[0.2]\}$ to points to its left.
Therefore, while processing $q'$,
no edge from points of $P$ to the right of $q'$ have been added. Hence,
$v(qq')$ does cross an existing edge.
\end{proof}

\begin{restatable}{lemma}{vert}\label{lem:vert}

For $i\in\{-t,\ldots, k\}$, 
let $q\in\Slab(R,p)\cap P$ s.t. $q$ lies above $\pi^u_i$. Let $q_1$ be the
right-neighbor of $q$ on $\pi^u_i$. If $qq_1$ is Delaunay but not valid, then
$v(qq_1)$ pierces a rectangle. In particular, $h(qq_1)$ does not pierce a rectangle
or cross an edge.
Similarly, let $q'\in \Slab(R,p)\cap P$ s.t. 
$q'$ lies below $\pi^{\ell}_i$,
$q'_1$ is the right-neighbor of $q'$ on $\pi^{\ell}_i$.
If $q'q'_1$ is Delaunay but not valid, then 
$v(q'q'_1)$ pierces a rectangle. In particular $h(q'q'_1)$ does not pierce a rectangle
or cross an edge.
\end{restatable}
\begin{proof}
We prove the case when $q$ lies above $\pi^u_{i}$. The other case follows
by an analogous argument.
Since $qq_1$ is not valid, either the horizontal
segment of $qq_1$ pierces a rectangle, or crosses an existing edge; or
$v(qq_1)$ pierces a rectangle since by Proposition \ref{prop:rightEdge},
$v(qq_1)$ does not cross an existing edge.

Let $q_0$ be left-adjacent to $q_1$ on $\pi^{u}_i$.
By Proposition \ref{prop:below}, $y(q)>\max\{y(q_0),y(q_1)\}$.
Hence $qq_1$ is of type $\upleft[0.2]$.
Suppose $h(qq_1)$ pierces a rectangle or crosses
an edge of type $\upleft[0.2]$. Then, there is a point $z$ that lies below
the $h(qq_1)$. But, $z$ cannot lie below the $h(q_0q_1)$, 
as that contradicts the fact that $q_0q_1$ is valid.
Hence, $z$ lies above $h(q_0q_1)$. This implies that
$z$ lies either in $R(q_0q_1)$ (if $y(q_0)<y(q_1)$), or $z$ lies in $R(qq_1)$. 
If $z\in R(q_0q_1)$, it contradicts the fact that $q_0q_1$ is Delaunay. Also, $z\not\in R(qq_1)$, as $qq_1$ is Delaunay by assumption.
Therefore, $h(qq_1)$ does not pierce a rectangle,
or crosses an edge of type $\upleft[0.2]$. If 
$h(qq_1)$ crossed an edge $e$ of type $\leftup[0.2]$, then either $qq_1$ is not
Delaunay, violating our assumption, or $e$ is not Delaunay, a contradiction.
\end{proof}

\begin{figure}[ht!]
\begin{center}
\includegraphics[scale=0.67]{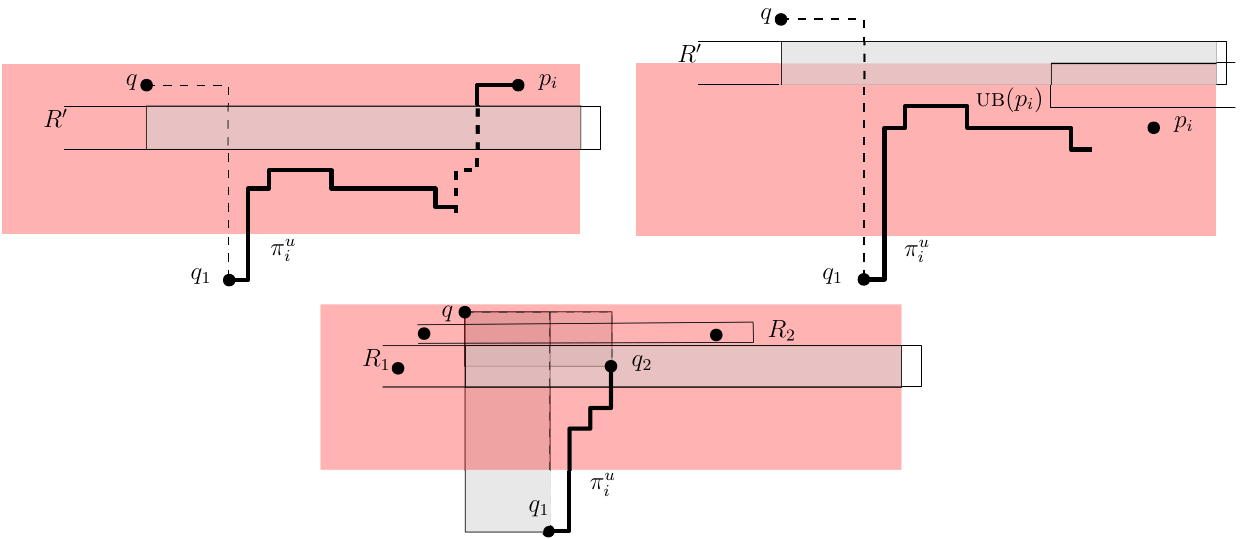}
\caption{Top left: $R'$ cannot be empty if $p_i$ is above $R'$, Top right: If $R'$ does not contain a point of $P$ between $x(q)$ and $x_+(R)$, then
$y(q)>y_+(\Strip_i)$, Bottom: Since $qq_1$ and $qq_2$ are Delaunay, then $R_2$ pierces $v(qq_1)$.}
\label{fig:cases}
\end{center}
\end{figure}


\begin{restatable}{lemma}{corridordeux}\label{lem:corridor2}
For $i\in\{-t,\ldots, k\}$, $P_i\subseteq \corridor_i$, where $P_i = P\cap\Strip_i$.
\end{restatable}
\begin{proof}
Suppose $P_i\setminus\corridor_i\neq\emptyset$. 
By Proposition \ref{obs:cor}, any such point either lies above $\pi^u_i$ or
below $\pi^{\ell}_i$. We assume the former. The latter
follows by an analogous argument.
Let
$P'_i = \{q\in \Slab(R,p): y(q)<y_+(\Strip_i)\mbox{ and } q \mbox{ lies above } \pi^u_i\}
$. It suffices to show that $P'_i=\emptyset$. For the sake of contradiction, suppose
$P'_i\neq\emptyset$.

We impose the following partial order on $P_i$: for $a,a'\in P_i$, $a\prec a'\Leftrightarrow
x(a)>x(a')\,\wedge\,y(a)<y(a')$. Let $q$ be a minimal element in
$P_i'$ according to $\prec$. In the following, when we refer to a minimal element, we impicitly
assume this partial order.

We show that there is a valid Delaunay edge between $q$ and a point $q'$ on $\pi^u_i$.
By assumption, $q$ lies above $\pi^u_i$. Let $q_0$ and $q_1$ denote, respectively,
the left- and right-neighbors of $q$ on $\pi^u_i$.


Since $q$ is minimal in $P'_i$, $qq_1$ is Delaunay. By Proposition \ref{prop:below},
it follows that $y(q)>y(q_0)$ and $y(q)>y(q_1)$. 
Since $q$ is not left-adjacent to $q_1$ on $\pi^u_i$, it implies $qq_1$ is not valid.

Since $qq_1$ is Delaunay but not valid, by Lemma \ref{lem:vert},
$v(qq_1)$ pierces a rectangle. Let $\mathcal{R}'$ denote the set of rectangles pierced by $qq_1$. 
Suppose $\exists R'\in\mathcal{R'}$ s.t. $R'[x(q_1),\min\{x(p), x_+(R')\}]\cap P=\emptyset$.
Then, we call $R'$ a \emph{bad} rectangle. Otherwise, we say that $R'$ is \emph{good}.
Now we split the proof into two cases depending on whether $\mathcal{R}'$ contains a bad
rectangle or not. In the two cases below, we use Proposition \ref{prop:xmonotone} 
that $\pi^u_i$ is $x$-monotone.

\smallskip\noindent
{\bf Case 1. $\mathcal{R'}$ contains a bad rectangle.}
Let $R'\in\mathcal{R}'$ be a bad rectangle. 
First, observe that $x_+(R')>x(p)$ as otherwise, 
$R'$ is not pierced by $v(qq_1)$.
Now, suppose $y(p_i)>y_+(R')$, where $p_i$ is the rightmost point in $\Strip_i$. 
We have that $x(q_1)<x(p_i)$, both $p_i$ and $q_1$ lie on $\pi^u_i$ and, 
$\pi^u_i$ is $x$-monotone. But, this implies $\pi^u_i$ pierces $R'$. But this is impossible as by construction
$\pi^u_i$ consists of valid Delaunay edges. Hence, since $R'$ is bad, we can assume that $y(p_i)<y_-(R')$. 
But the definition of the upper barrier implies that 
$y_+(R')\ge y_+(\ub(p_i))$. Since $v(qq_1)$ pierces
$R'$, it implies $y(q)>y_+(R')$, and hence $y(q)>y_+(\ub(p_i))$.
But, this contradicts the assumption that $q\in\Strip_i$, since
by Proposition \ref{prop:strip}, $y_+(\ub(p_i))\ge y_+(\Strip_i)$
and hence $y(q)>y_+(\Strip_i)$.

\smallskip\noindent
{\bf Case 2. All rectangles in $\mathcal{R}'$ are good.} 
Let $R_1=\arg\max\{y_-(R'):R'\in\mathcal{R}'\}$.
Let $q_2$ be the leftmost point in $R_1$ s.t. $x(q_2)>x(q_1)$.
Since $q_2$ lies to the right, and below
$q$, $q_2\prec q$. Since $q$ is a minimal element in $P'_i$, 
it implies that $q_2$ lies on or below $\pi^u_i$. We claim that
$q_2$ cannot lie below $\pi^u_i$. Suppose it did. Let $q_2'$ be the left-neighbor
of $q_2$ on $\pi^u_i$. Then, $x(q_2')<x(q_2)$. 
Since $q_2$ is the leftmost point of $R_1$ to the
right of $v(qq_1)$, then either $y(q'_2)<y_-(R_1)$, or $y(q'_2)>y_+(R_1)$. However,
in either case, we obtain that $\pi^u_i$ must cross $R_1$ between $q_2$ and $q_1$,
which implies that $\pi^u_i$ pierces $R_1$, as $\pi^u_i$ is $x$-monotone, and
the edges in $\pi^u_i$ are of the form $\{\upleft[0.2],\leftup[0.2]\}$, a contradiction.

Since $q$ is minimal, $qq_2$ is Delaunay.
By Lemma \ref{lem:vert}, the only reason
$qq_2$ is not valid is that $v(qq_2)$ pierces a rectangle. But, any such rectangle
$R_2$
is also pierced by $v(qq_1)$, as $qq_2$ is Delaunay. But, this implies
$y_-(R_2)>y_-(R_1)$, contradicting the choice of $R_1$. Therefore, $qq_2$ is a valid
Delaunay edge. Now, the only reason that $q$ is not the left-adjacent point
of $q_2$ on $\pi^u_i$ then, is that $q'_2$, the left-adjacent point of $q_2$ on $\pi^u_i$
lies in $\Slab(R,p)$, but above $q$, 
i.e., $y(q'_2)>y(q)$, as the construction of $\pi^u_i$
dictates that we choose the left-adjacent point with highest $y$-coordinate
that lies in $\Slab(R,p)$. Showing that this leads to a contradiction completes
the proof.

So suppose $y(q'_2)>y(q)$, then $y(q'_2)>y_+(R_1)$. Further, 
by Proposition \ref{prop:deledges}, $x(q'_2)>x(q)$. Again this implies
the $x$-monotone curve $\pi^u_i$ cannot contain both $q'_2$ and $q_1$
without piercing $R_1$. Hence, $y(q'_2)<y_-(R_1)<y(q)$, but this contradicts
the choice of $q'_2$ as the left-adjacent point of $q_2$ on $\pi^u_i$ since $qq_2$ is a valid Delaunay edge with $y(q'_2)<y(q_2)<y(q)$.
See Figure \ref{fig:cases} for the different cases in this proof.
\end{proof}

The key property of a corridor is that if the upper or lower path of a corridor crosses a rectangle $R'$,
then there must be a point of $R'\cap P$ that lies on that path of the corridor. Using this, we can show
that any point in the strip has an adjacent point to its right in $G$ that lies in the corridor.
This implies that every point in a strip has a path to the rightmost point in the strip
that lies entirely in its corresponding corridor.

\begin{lemma}
\label{lem:corridor3}
For each $q\in \corridor_i$, there is a path $\pi(q,p_i)$ between $q$ and $p_i$ that lies in $\corridor_i$, 
where $p_i\in P_i$ is the rightmost point in $\Strip_i$.
\end{lemma}
\begin{proof}
If $q$ lies on the upper path $\pi^u_i$ or the lower path $\pi^{\ell}_i$
defining $\corridor_i$, the lemma is immediate. So we can assume by Proposition \ref{obs:cor} that
$q$ lies below $\pi^u_i$, and above $\pi^{\ell}_i$. Suppose the lemma is false. 
Let $q$ be the rightmost point of $\corridor_i$ that does not have a path to $p_i$ lying in $\corridor_i$.
To arrive at a contradiction, it is enough to show that $q$ is adjacent to a point $q'\in \corridor_i$ that
lies to the right of $q$. 

Starting from $\ell_q$, the vertical line through $q$, sweep to the right until the first point $r$ that lies on both $\pi^{u}_i$ and
$\pi^{\ell}_i$. Such a point exists since $p_i$ lies on both $\pi^{u}_i$ and $\pi^{\ell}_i$.
Let $Q'_i$ denote the set of points in $\corridor_i$ whose $x$-coordinates lie between $x(q)$ and $x(r)$.
This set is non-empty as it contains $q$ and $r$. 
Hence, either $Q^+_i = \{q'\in Q'_i: y(q')>y(q)\}\neq\emptyset$, or $Q^-_i=\{q'\in Q'_i:y(q')<y(q)\}\neq\emptyset$.
If both are non-empty, let $Q_i$ denote the set that contains a point with smallest $x$-coordinate. Otherwise,
we let $Q_i$ denote the unique non-empty set. Assume $Q_i=Q^+_i$. An analgous argument holds when $Q_i=Q^-_i$.

Define a partial order on $Q_i$, where for $a,b\in Q_i$, $a\prec b\Leftrightarrow x(a)<x(b) \mbox{ and } y(a)<y(b)$.
Let $Q^{\min}_i=(q_1,\ldots, q_t)$ denote the sequence of minimal elements of $Q_i$ ordered linearly such that $y(q_k)>y(q_j)$ for $k < j$. It follows that $x(q_k)<x(q_j)$ for $k<j$. Observe that $qq_i$ is Delaunay for $i=1,\ldots, t$ by the minimality of $q_i$.
Our goal is to show that $qq_i$ is a valid Delaunay edge for some $i\in\{1,\ldots, t\}$.
We start with the following claim that $v(qq_t)$ and $h(qq_1)$ do not pierce a rectangle in $\mathcal{R}$, 
or cross an edge of $G$. 

\begin{claim}
\label{claim:qqt}
$h(qq_1)$ does not pierce a rectangle in $\mathcal{R}$ or cross an edge of $G$, and
$v(qq_t)$ does not pierce a rectangle in $\mathcal{R}$ or cross an edge of $G$.
\end{claim}
\begin{proof}
Suppose $h(qq_1)$ pierced a rectangle $R'$. Since $\pi^u_i$ consists of valid Delaunay
edges, and the choice of $Q_i$, $R'$ contains a point $a$ that lies in $\corridor_i$. Since
$h(qq_1)$ pierces $R'$, $x(q)<x(a)<x(q_1)$. If $y(a)<y(q_1)$, then it contradicts the fact that
$q_1$ is minimal, and if $y(a)>y(q_1)$, it contradicts the fact that $q_1$ is the minimal
element with highest $y$-coordinate. A similar argument shows that $h(qq_1)$ does not
cross an edge of $G$.

If $v(qq_t)$ pierced a rectangle $R'\in\mathcal{R}$,
then $R'$ has a point to the right of $v(qq_t)$. Further, by Proposition \ref{prop:xmonotone},
$\pi^u_i$ and $\pi^{\ell}_i$ are $x$-monotone paths and by construction, they consist
of valid Delaunay edges
meeting at $r$. 
If $r=q_t$ and $v(qq_t)$ pierced a rectangle, since $qq_t$ is Delaunay, and 
the edges are of type $\{\upleft[0.2],\leftup[0.2]\}$, it implies that
the left-adjacent point of $r$ on $\pi^u_{i}$ or $\pi^{\ell}_i$ is not a valid Delaunay edge. Hence,
we can assume $q_t\neq r$. Again, since the edges of $\pi^{u}_i$ and $\pi^{\ell}_i$ are valid Delaunay
edges, 
it implies that $R'$ has a point $a$ s.t. $x(q_t)<x(a)\le x(r)$, and $a$ lies in $\corridor_i$. 
Since $v(qq_t)$ pierces $R'$,
$y(a)<y(q_t)$. But this contradicts the fact that $q_t$ is the point in $Q^{\min}_i$ with the smallest $y$-coordinate.
Therefore, $v(qq_t)$ can not pierce any rectangle.
Since $qq_t$ is Delaunay, by Proposition \ref{prop:rightEdge}, $v(qq_t)$ does not cross an edge of $G$.
\end{proof}

We now define a point $x_1$ on the $x$-axis and a point $y_1$ on the $y$-axis as follows:
\begin{align*}
x_1 &= \min\left\{\{x_+(R'):R'\mbox{ pierced by } h(qq_t)\},\{v(e): e\mbox{ crosses } h(qq_t)\}\right\}\\
y_1 &= \min\left\{\{y_+(R'):R'\mbox{ pierced by } v(qq_1)\},\{h(e): e\mbox{ crosses } v(qq_1)\}\right\}
\end{align*}
By Claim \ref{claim:qqt} and the assumption that $qq_t$ and $qq_1$ are not valid edges, it follows that
$x_1$ and $y_1$ exist.  
We argue when $x_1$ and $y_1$ correspond to $x_+(R')$ and
$y_+(R'')$, respectively for rectangles $R',R''\in\mathcal{R}$. 
If they were instead defined by the vertical/horizontal side of edges, the arguments are similar.

\begin{figure}[ht!]
\begin{center}
\includegraphics[width=2.5in]{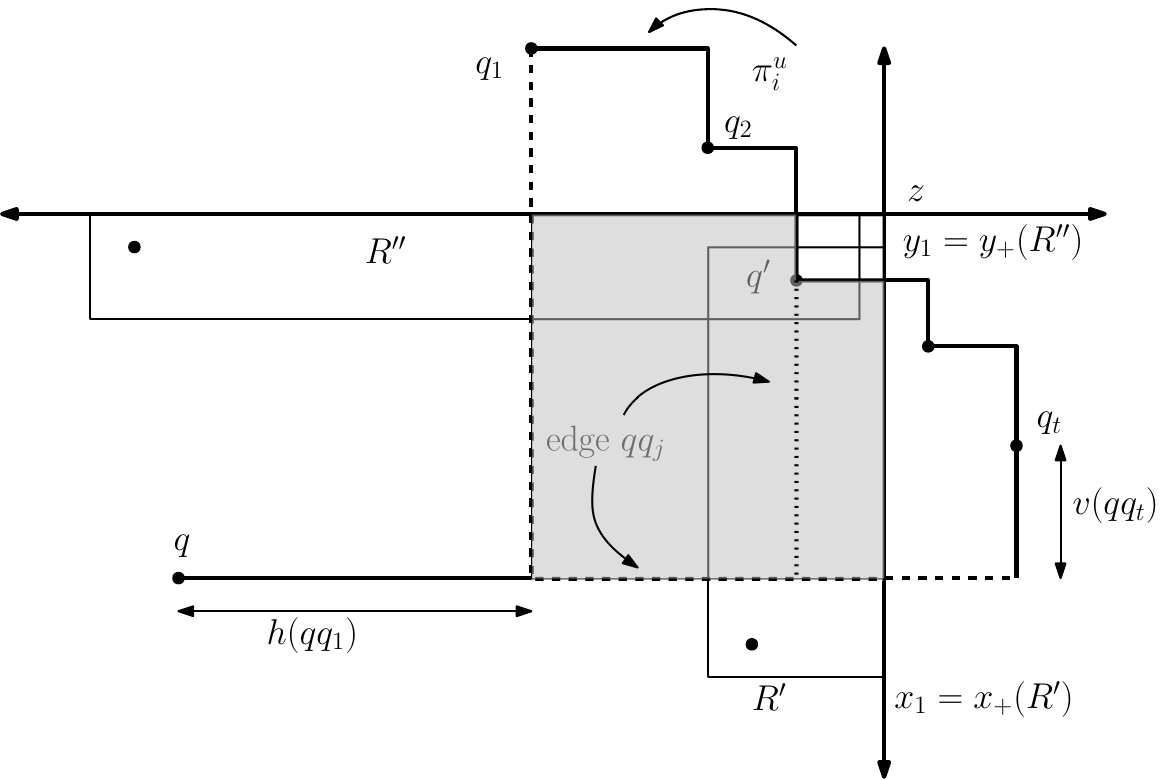}
\caption{The edge $qq'$ is a valid Delaunay edge.}
\label{fig:godeedge}
\end{center}
\end{figure}

Observe that $x_-(R'')<x(q)$, while $x_-(R')>x(q)$, and $y_-(R'')>y(q)$ and $y_-(R')<y(q)$.
Now, from the fact that the rectangles are non-piercing, it implies that 
either $x_+(R'')<x_1$ or $y_+(R')<y_1$.
Suppose wlog, the former is true.
Since $R'$ is pierced by $h(qq_t)$ and $\pi^u_i$ consists of valid Delaunay edges, 
there are points in $R'$ that lie in $\corridor_i$, and these
points lie below $y_1$. 

Let $z$ denote the intersection of the vertical line through $x_1$ and the hoizontal line through
$y_1$. By the argument above, the rectangle with diagonal $qz$ contains points of $P$, and
hence a point $q'\in Q^{\min}_i$. We claim that $qq'$ is a valid Delaunay edge. To see this, 
note that $h(qq')$ does not pierce a rectangle in $\mathcal{R}$ as such a rectangle contradicts
the definition of $x_1$. If $v(qq')$ pierced a rectangle, such a rectangle $\tilde{R}$ must have
$y_+(\tilde{R})<y_1$, as $qq'$ is Delaunay. This contradicts the choice of $y_1$. 
Therefore, $qq'$ is a valid Delaunay edge.
\end{proof}

The lemma below follows the description in the proof strategy at the start of this section.

\begin{restatable}{lemma}{monpathslab}\label{lem:monpathslab}

For a rectangle $R$ and point $p\in R$, after Algorithm~\ref{alg:prim} has processed point $p$,
the points in $\Slab(R,p)$ induce a connected subgraph, all of whose edges lie in $\Slab(R,p)$.
\end{restatable}
\begin{proof}
Let $G[\Slab(R,p)]$ denote the induced subgraph of $G$ on the points in $\Slab(R,p)$.
By Condition (ii) of Lemma \ref{lem:strips}, 
since $\Slab(R,p)\subseteq \cup_{i=-t}^k \Strip_i$, each point in $P\cap\Slab(R,p)$ is contained in $\cup_{i=-t}^k\Strip_i$.
If the statement of the lemma does not hold, consider an extremal strip, i.e., the smallest positive index, or 
largest negative index of a strip such that it contains a point $q$ that does
not lie in the connected component of $G[\Slab(R,p)]$ containing $p$.
Assume without loss of generality that $i \ge 0$. An analogous argument holds
if $i < 0$.
By Lemma \ref{lem:corridor2}, $q\in\corridor_i$, and by Lemma \ref{lem:corridor3}, $q$ has a path $\pi_1$ to $p_i$, the rightmost point in 
$\corridor_i$ that lies entirely in $\corridor_i$. 
By Condition (ii) of Lemma \ref{lem:strips}, there is a point
$q'\in \Strip_i\cap\Strip_{i-1}\cap P$. By Lemma \ref{lem:corridor2},
$q'\in\corridor_i$, and by Lemma \ref{lem:corridor3}, there is a path
$\pi_2$ between $q'$ and $p_i$. Since $q'\in\Strip_{i-1}$, $q'$ lies in the
same connected component as $p$ in $G[\Slab(R,p)]$, and hence there is a path $\pi'$ from $q'$ to $p$ in $G[\Slab(R,p)]$. Concatenating $\pi_1, \pi_2$
and $\pi'$ we obtain a path $\pi$ from $q$ to $p$ that lies in $\Slab(R,p)$.
\end{proof}

We now argue that if $p$ is the rightmost point in a rectangle $R$, then $\Piece(R,p)$ consists of a single slab.

\begin{restatable}{lemma}{lastpt}\label{lem:lastpt}

If $p$ is the last point in $R$ according to the x-coordinates of the points, then $\Piece(R,p)$ consists of a single slab.
\end{restatable}
\begin{proof}
Assume for the sake of contradiction that $\upb(R,p)$ exists. By definition of $\upb(R,p)$,
there are two points $a,b\in \upb(R,p)$, such that $x(a) < x_-(R) < x(p) < x_+(R) < x(b)$, as $p$ is the last point in $R$.
But this implies $\upb(R,p)$ is pierced by $R$, a contradiction. Therefore, $\upb(R,p)$ does not exist. Similarly, $\lpb(R,p)$ does not exist,
and hence $\Piece(R,p)$ consists of a single slab.
\end{proof}

\begin{theorem}
\label{thm:mainthm}
Algorithm~\ref{alg:prim} constructs a planar support.
\end{theorem}
\begin{proof}
By construction, the edges of the graph $G$ constructed by Algorithm~\ref{alg:prim} are valid Delaunay edges of type $\{\upleft[0.2],\leftup[0.2]\}$. To obtain a plane embedding, we replace each edge $e=\{p,q\}$ by the diagonal of the
rectangle $R(pq)$ joining $p$ and $q$. We call these the \emph{diagonal edges}. It is clear that no diagonal edge
pierces a rectangle. If two diagonal edges cross, then it is easy to check that either the corresponding edges cross, 
or they are not Delaunay. 
For a rectangle $R\in\mathcal{R}$, let $p$ be the last point in $R$.
Lemma~\ref{lem:lastpt} implies that there is only one slab, namely $R$, and 
Lemma~\ref{lem:monpathslab} implies $\textsc{Slab}(R,p)$ is connected. Since $R$ was arbitrary, this implies
Algorithm~\ref{alg:prim} constructs a support. 
\end{proof}




\subsection{Implementation}
\label{sec:implementation}

In this section, we show that Algorithm~\ref{alg:prim} can be implemented to run in $O(n\log^2n + (m+n)\log m)$ time with appropriate data structures,
where $|\mathcal{R}|=m$, and $|P|=n$. 
At any point in time, our data structure maintains a subset of points that lie to the left of the sweep line $\ell$. It also maintains for each rectangle $R$ intersecting $\ell$, the interval $[y_-(R), y_+{R}]$ corresponding to $R$. When the sweep line arrives at the left side of a rectangle, the corresponding interval is inserted into the data structure. The interval is removed from the data structure when the sweep line arrives at the right side of the rectangle. Similarly, whenever we sweep over a point $p$, we insert it into the data structure. In addition, we do the following when the sweep line arrives at a point $p$:
\begin{enumerate}[1.]
    \item \label{step1} Find the upper and lower barriers at $p$.

    \item \label{step2}  Query the data structure to find the set $Q$ of points $q$ which $i)$ lie to the left of $p$ and between the upper and lower barriers at $p$ (orthogonal range query) so that $ii)$ $qp$ is a Delaunay edge. 
    
    \item \label{step3} We add the edge $qp$ for every $q \in Q$ to our planar support. For each edge we add, we remove the points in the data structure that are {\em occluded} by the edges. These are the points whose $y$-coordinates lie in the range corresponding to the vertical side of the $L$-shape for $qp$. 
\end{enumerate}

Our data structure is implemented by combining three different existing data structures. For Step~\ref{step1}, we use a balanced binary search tree $\cT^u_1$ augmented so that it can answer range minima or maxima queries. For any rectangle $R$ intersecting the sweep line $\ell$,  let $(y_1, y_2)$ denote the interval corresponding to the projection of $R$ on the $y$-axis. $\cT_1^u$ stores the key-value pair $(y_1, y_2)$ with $y_1$ as the key and $y_2$ as the value. To find the upper barrier at a point $p = (x,y)$ we need to find the smallest value associated with keys that are at least $x$. If we augment a standard balanced binary search tree so that at each node we also maintain the smallest value associated with the keys in the subtree rooted at that node, such a query takes $O(\log m)$ time. An analogous search tree $\cT^b_1$ is used to find the lower barrier at any point.

To implement Step~\ref{step2}, we use a dynamic data structure $\cT^b_2$ due to Brodal~\cite{brodal2011dynamic} which maintains a subset of the points to the left of $\ell$ and 
can report points in any query rectangle $Q$ that are not dominated by any of the other points in time $O(\log^2 n + k)$ where $k$ is the number of reported points. We say that a point $u$ is dominated by a point $v$ if both $x$ and $y$ coordinates of $u$ are smaller than those of $v$. 
The data structure also supports insertions or deletions of points in $O(\log^2 n)$ time.
When the sweep line arrives at a point $p$, we can use 
$\cT^b_2$ to find all points $q$ that lie to the left of $p$ and below $p$ so that the edge $qp$ is a Delaunay edge (as $qp$ of shape $\leftup[0.2]$ is Delaunay iff there is no other point in the range below and to the left of $p$ that dominates $q$). An analogous data structure $\cT^u_2$ is used to find the points $q$ which lie above and to the left of $p$ so that $qp$ (of shape $\upleft[0.2]$) is Delaunay. 

To implement Step~\ref{step3}, we use a dynamic 1D range search data structure $\mathcal{T}_3$ which also stores a subset of the points to the left of $\ell$, supports insertions and deletions in $O(\log n)$ time and can report in $O(\log n + k)$ time the subset of stored points that lie in a given range of $y$-coordinates (corresponding the vertical side of each added edge),  where $k$ is the number of points reported.
The points identified are removed from $\cT_2^u, \cT_2^b$ and $\mathcal{T}_3$.

By the correctness of Algorithm~\ref{alg:prim} proved in Section \ref{sec:correctness}, 
at any point in time, the current graph is a support for the set of rectangles that lie completely to the left of the sweep line. Thus, if the sweep line $\ell$ is currently at a point $p$ and $q$ is a point to the left of $\ell$, the only rectangles that $qp$ may discretely pierce are those that intersect $\ell$.
A simple but important observation is that if $qp$ is Delaunay then $qp$ pierces a rectangle iff the vertical portion of $L$-shape forming the edge $pq$ pierces the rectangle. To see this note that the horizontal portion of the $L$-shape cannot pierce any rectangle since such a rectangle would not intersect $\ell$. The $L$-shape also cannot (discretely) pierce a rectangle containing the corner of the $L$-shape since then the edge $qp$ would not be a Delaunay edge. Thus, in order to avoid edges that pierce other rectangles, it suffices to restrict $q$ to lie between the upper and lower barriers at $p$.  Thus Step~\ref{step1} above ensures that edges found in Step~\ref{step2} don't pierce any of the rectangles. Similarly, 
Step~\ref{step3} ensures that the edges we add in Step~\ref{step2} don't intersect previously added edges.

The overall time taken by the data structures used by Step~\ref{step1} is $O( (m+n) \log m)$ since it takes $O(\log m)$ time to insert or delete the key-value pair corresponding to any of the $m$ rectangles, and it takes $O(\log m)$ time to query the data structure for the upper and lower barriers at any of the $n$ points. The overall time taken by the data structure in Step $2$ is $O(n \log^2 n)$ since there are most $O(n)$ insert, delete, and query operations, and the total number of points reported in all the queries together is $O(n)$. The overall time taken by the data structure in Step~\ref{step3} is $O(n \log n)$ we only add $O(n)$ edges in the algorithm and the query corresponding to each edge takes $O(\log n)$ time. Each of the reported points is removed from the data structure but since each point is removed only once, the overall time for such removals is also $O(n\log n)$. The overall running time of our algorithm is therefore $O(n\log^2 n + (m+n) \log m)$.

\section{General families of rectangles}
\label{sec:Application}
In this section, we construct an embedded (not-necessarily planar) support graph $G$ for an arbitrary family 
$\mathcal{R}$ of non-piercing rectangles using the results in the previous section.
If $\mathcal{R}$ can be partitioned into $t$ families of non-piercing rectangles, then $G$ is the
union of $t$ plane graphs.

We show that there is a polynomial time algorithm to partition a family $\mathcal{R}$ of axis-parallel rectangles
into a minimum number of parts, each of which is a collection of non-piercing rectangles.

To that end, we construct the \emph{piercing} graph $H=(\mathcal{R},E)$ on the rectangles, where for each $R_1, R_2\in\mathcal{R}$, $\{R_1, R_2\}\in E$ if
$R_1$ and $R_2$ (geometrically) pierce, i.e., $R_1$ and $R_2$ intersect, but neither contains
a corner of the other. 
We claim that $H$ is a \emph{comparability graph}. A graph on vertex set $V$ is a comparability graph if
$H$ can be obtained from the transitive closure of a partial order on $V$ by forgetting directions.
Comparability graphs are perfect, and it is known that they can be colored optimally in polynomial time \cite{GOLUMBIC1980105}. 
In our context, this means that we assign colors to the rectangles so that
piercing rectangles receive distinct colors.

\begin{theorem}
\label{thm:comparability}
For a set $\mathcal{R}$ of axis-parallel rectangles in the plane, the piercing graph $H$ is a comparability graph.
\end{theorem}
\begin{proof}
For a rectangle $R\in\mathcal{R}$, let $x(R)$ and $y(R)$ denote respectively, the projections of $R$ on the $x$ and $y$ axes.
Two rectangles $R_1$ and $R_2$ pierce if and only if $x(R_1)\subseteq x(R_2)$, and $y(R_2)\subseteq y(R_1)$ or vice versa. Consider the containment
order $\prec_x$, where $R_1\prec_x R_2$ if and only if $x(R_1)\subseteq x(R_2)$. Similarly, consider the containment order
$\prec_y$, where $R_1 \prec_y R_2$ if and only if $y(R_2)\subseteq y(R_1)$ (note that the order is reversed). That is, the partial order on the rectangles induced
by their $y$-projection is the reverse of the containment order. Now, $R_1$ and $R_2$ are adjacent in $G$ if and only if 
$R_1\preceq_x R_2$ and $R_1\preceq_y R_2$. Since the intersection of two partial orders is a partial order, we obtain
a partial order $\prec_P$ where $R_1\prec_P R_2$ iff $R_1$ and $R_2$ pierce. We obtain $H$ from $\prec_P$ by completing
the transitive closure of $\prec_P$ and forgetting the directions of the edges.
\end{proof}

From Theorem~\ref{thm:comparability}, $H$ is a comparability graph.
Since comparability graphs are perfect \cite{GOLUMBIC198051,GOLUMBIC1980105}, the chromatic number and
a coloring with the fewest number of colors can be computed in polynomial time.

We now use the fact that there is a planar support for the intersection graph of axis-parallel rectangles
to obtain an algorithm with small approximation factor for the \emph{Hitting Set} problem for
any family of axis-parallel rectangles whose piercing graph has small chromatic number.
In our context, the Hitting Set problem has input a set $P$ of points and a set $\mathcal{R}$ of
axis-parallel rectangles. The objective is to select a smallest subset $P'\subseteq P$ s.t. $P'\cap R\neq\emptyset$
for any $R\in\mathcal{R}$. Note that for this result, we only require the existence of a planar support for 
points with respect to non-piercing rectangles, and we do not need to explicitly construct them. Such a result was already
known, and was proved by \cite{PyrgaR08}, which was used by Musafa and Ray \cite{mustafa2010improved}
to show that a simple \emph{local search} algorithm yields a PTAS for a more general family of \emph{pseudodisks}. 
Our contribution then, 
is only to use the observation in Theorem \ref{thm:comparability} to show that the piercing graph is a comparability graph,
and if it has small chromatic number, this implies an approximation algorithm with small factor for the Hitting Set problem.

\begin{theorem}
\label{thm:param}
Let $\mathcal{R}$ be a set of axis-parallel rectangles in the plane
such that the chromatic number $\chi(H)$ of the piercing graph $H$ is at most $k$.
Then, for any $\epsilon > 0$, there is a $(k+\epsilon)$-approximation algorithm for the 
Hitting Set problem on $(\mathcal{R}, P)$.
\end{theorem}
\begin{proof}
The piercing graph $H$ of $\mathcal{R}$ can be computed in polynomial time. By Theorem~\ref{thm:comparability}, $H$ is a comparability graph. Since $\chi(H)\le k$ and a coloring $H$ with at most $k$ colors can be computed in polynomial
time \cite{GOLUMBIC1980105}, we can partition $\mathcal{R}$ into at most $k$ families
$\mathcal{R}_1,\ldots, \mathcal{R}_k$ s.t. each $\mathcal{R}_i$
consists of non-piercing rectangles. 

To compute the hitting set for rectangles $\mathcal{R}$ and points $P$, we apply a PTAS for the Hitting Set Problem \cite{mustafa2010improved} with rectangles $\mathcal{R}_i$ and point set $P$. We return the
union of the solutions for the $k$ Hitting Set problems, whose size we denote by $S$. 
Let \OPT{} denote the optimal size of a hitting set for the input rectangles with points $P$.
Let $\OPT{}_i$ and $S_i$ denote respectively, the size of an optimal hitting set for the rectangles in $\mathcal{R}_i$, and the size of a hitting set 
returned by the PTAS for $\mathcal{R}_i$. Consider $\epsilon'=\epsilon/k$, where each PTAS returns an $\epsilon$-approximate solution. Then,%
\begin{align*}
S &\le S_1 + S_2 + \ldots + S_k\\
  &\le \left(1+\epsilon'\right)\OPT{}_1 + \left(1+\epsilon'\right)\OPT{}_2 + \ldots + \left(1+\epsilon'\right)\OPT{}_k\\
  &\le \left(k+\epsilon \right) \OPT{}.%
\end{align*}%
\end{proof}%
\begin{theorem}
Let $\mathcal{R}$ be a collection of axis-parallel rectangles, and $P$ be a set of points 
such that the piercing graph has chromatic number $k$. Then, there is a collection of $k$ planar graphs
whose union is a support for the hypergraph $(P,\mathcal{R})$.
\end{theorem}
\begin{proof}
The piercing graph of $\mathcal{R}$ is a comparability graph by Theorem \ref{thm:comparability}. Since we can compute
the chromatic number and a coloring with the smallest number of colors in polynomial time for such graph classes \cite{GOLUMBIC1980105}, we can partition $\mathcal{R}$ into $k$ color classes $\mathcal{R}_1, \ldots, \mathcal{R}_k$
such that each color class consists of a set of non-piercing rectangles. 
By Theorem~\ref{thm:mainthm}, for each color class $\mathcal{R}_i$ and points contained in the rectangles in $\mathcal{R}_i$,
we can construct a planar support. The union of the support graphs yields a support for $\mathcal{R}$ as a union of
$k$ planar graphs.
\end{proof}

\section{Conclusion}
\label{sec:conclusion}
We presented a fast algorithm to compute a plane support graph for non-piercing axis-parallel rectangles that runs in $O(n\log^2 n + (n+m)\log m)$. Although the algorithm itself is simple, the fast implementation relies on sophisticated orthogonal range searching data structures of Brodal et al., \cite{brodal2011dynamic}, and the proof of correctness itself turned out to be surprisingly complex. Besides the natural question of whether we can improve the running time of our algorithm to $O((n+m)\log (m+n))$, there are several open questions. Here, we suggest a few such questions motivated by the work of 
Raman and Ray \cite{RR18} who obtained polynomial time algorithms to construct \emph{dual} and
\emph{intersection} supports for non-piercing regions. In our context, these questions translate into the following: $(i)$ Given a set of points $P$ and axis-parallel non-piercing rectangles $\mathcal{R}$ in the plane, obtain a fast algorithm to construct a plane support graph $Q^*=(\mathcal{R},E)$ on $\mathcal{R}$ such that for each point $p\in P$, the rectangles containing $p$ induce a connected subgraph in $Q^*$, and $(ii)$ Given two sets $\mathcal{R}$, $\mathcal{S}$, each of which is a collection of non-piercing rectangles, obtain a fast algorithm to construct a plane support graph $\tilde{Q}$ on $\mathcal{R}$ such that for each
$S\in\mathcal{S}$, the induced graph of $\tilde{Q}$ on 
$\{R\in\mathcal{R}: R\cap S\neq\emptyset\}$ is connected. Another interesting problem is to obtain fast algorithms when, instead of non-piercing axis-parallel rectangles, we have non-piercing convex sets given by either a membership, or separation oracle. In this case, it is relatively easy to show that a maximal set of edges (straight segment joining a pair of points) which do not discretely pierce any of the convex sets and are pairwise non-crossing, form a planar support. This yields a polynomial time algorithm but it is not clear whether a faster algorithm exists.

\bibliography{ref}
\end{document}